\documentclass{article}
\usepackage{latexsym,graphicx,amsthm}
\newtheorem{lemma}{Lemma}
\newtheorem{theorem}{Theorem}
\usepackage[numbers,sort&compress]{natbib}

\begin{document}

\date{October 24, 2001}

\title{\textbf{Brooks' Vertex-Colouring\\ Theorem in Linear Time} \thanks{Technical Report CS-AAG-2001-05, Basser Department of Computer Science, The University of Sydney, 2001. Completed as an undergraduate research project at The University of Sydney (Sydney, Australia)  and McGill University (Montr\'eal, Canada).}}

\author{Bradley Baetz \qquad David R. Wood\footnotemark[2]}

\footnotetext[2]{School of Mathematical Sciences, Monash University, Melbourne Australia (\texttt{david.wood@monash.edu}). Research supported by the Australian Research Council.}

\maketitle

\begin{abstract} Brooks' Theorem  [R.\ L.\ Brooks, On Colouring the Nodes of a
Network, \emph{Proc.\ Cambridge Philos.\ Soc.} \textbf{37}:194-197, 1941]
states that every graph $G$ with maximum degree $\Delta$, has a
vertex-colouring with $\Delta$ colours, unless $G$ is a complete graph or an
odd cycle, in which case $\Delta+1$ colours are required. Lov{\'a}sz [L.\
Lov{\'a}sz, Three short proofs in graph theory, \emph{J.\ Combin.\ Theory Ser.}
\textbf{19}:269-271, 1975] gives an algorithmic proof of Brooks' Theorem. 
Unfortunately this proof is missing important details and it is thus unclear
whether it leads to a linear time algorithm. In this paper we give a complete
description of the proof of Lov{\'a}sz, and we derive a linear time algorithm
for determining the vertex-colouring guaranteed by Brooks' Theorem. 
\end{abstract}

%\bigskip

\section{Introduction}

Let $G=(V,E)$ be a simple graph with maximum degree $\Delta$. Undefined
graph-theoretic terms can be found in~\cite{CL96}. A \emph{vertex-colouring} of
$G$ is a mapping from the vertex set of $G$ to some set of colours such that
adjacent vertices receive different colours. For convenience we take the set of
colours to be the positive integers $\{1,2,\dots\}$. A graph is said to be
$k$-\emph{colourable} if it has a vertex-colouring with at most $k$ colours.
Minimising the number of colours in a vertex-colouring of a given graph is a
fundamental problem in algorithmic graph theory with applications in register
allocation for example. Unfortunately, determining if a given graph is
$k$-colourable is NP-complete~\cite{Karp72}. The sequential greedy algorithm,
which chooses for each vertex $v$ in turn, the minimum colour not used by a
neighbour of $v$, will use at most $\Delta+1$ colours, since at each vertex $v$
there is at most $\Delta$ different colours assigned to the neighbours of $v$.
Brooks~\cite{Brooks41} proved the following improvement to this result.

\begin{theorem}[Brooks' Theorem~\cite{Brooks41}]
Every graph $G$ with maximum degree $\Delta$, has a vertex-colouring with $\Delta$
colours, unless $G$ is a complete graph or an odd cycle, in which case $\Delta+1$ colours
are required.
\end{theorem}

We say a vertex-colouring of graph $G$ with maximum degree $\Delta$ is a
\emph{Brooks-colouring} if the number of colours is at most $\Delta$, or $\Delta+1$ if
$G$ is a complete graph or an odd cycle. 

The original proof by Brooks leads to a quadratic time algorithm for
calculating a Brooks-colouring. Since then  Ponstein~\cite{Ponstein69} and
Lov{\'a}sz~\cite{Lovasz75} (also see Bryant~\cite{Bryant96}) describe
algorithmic proofs of Brooks' Theorem. However, the running time of the
resulting algorithms are not analysed, and in the proof by
Lov{\'a}sz~\cite{Lovasz75}, many important details are omitted.
In this paper, we give a complete description of the proof of Brooks' Theorem
due to Lov{\'a}sz~\cite{Lovasz75}, and derive a linear time algorithm for
computing a Brooks-colouring.

\section{The Details}

We start with the following well-known result,
which can be proved by performing a pre-order traversal of the block-cut-forest
of the given graph, and possibly swapping two colours in each biconnected
component.

\begin{lemma}
\label{lem:NonBiconnected}
Given $k$-colourings of the biconnected components of a graph $G=(V,E)$, a $k$-colouring
of $G$ can be determined in $O(|V|+|E|)$ time.\qed
\end{lemma}
%\begin{proof}
%Consider the biconnected components of $G$ in turn, according to a
%pre-order numbering of the block-cut tree of $G$. 
%Each biconencted component $B$ of $G$ has at most one cut vertex $v$ in some
%lower-numbered  biconnected component $B'$. 
%Suppose $v$ is coloured $i$ in the colouring of $B$ and $v$ is coloured $j$ in the
%colouring of $B'$.
%If $i\ne j$ then re-colour the vertices in $B$ previously coloured $i$ by $j$, and
%re-colour the vertices in $B$ previously coloured $j$ by $i$.
%Clearly this produces a valid vertex-colouring of $B$.
%Doing so for all biconnected components of $G$ results in a colouring of $G$ with all
%vertices receiving one colour.
%The number of colours is simply the maximum number of colours used in a vertex-colouring
%of a biconnected component.
%This procedure can be implemented in $O(|V|+|E|)$ using standard linear time depth-first
%algorithms.
%\end{proof}

Lemma~\ref{lem:NonBiconnected} implies that we need only describe a linear time
algorithm  for determining a Brooks-colouring in the case of a biconnected
graph. 

\begin{lemma}
\label{lem:SequentialColouring}
If a graph $G$ with maximum degree $\Delta$ contains two vertices $a$ and $b$
at distance 2 such that $G\setminus\{a,b\}$ is connected, then a
$\Delta$-colouring of $G$ can be determined in $O(|V|+|E|)$ time.
\end{lemma}
\begin{proof}
Let $(v_1,v_2,\dots,v_n)$ be an ordering of the vertices of
$G\setminus\{a,b\}$ such that $v_1$ is a vertex adjacent to both $a$ and $b$,
and for every $i\geq2$, the vertex $v_i$ has at least one neighbour $v_j$ with
$j<i$. Since $G\setminus\{a,b\}$ is connected, such an ordering can be
determined by a depth-first search of $G\setminus\{a,b\}$. Let the colour of
$a$ and $b$ be $1$. This is valid since $a$ and $b$ are not adjacent. Now, for
each $i$, $i=n,n-1,\dots,1$, colour the vertex $v_i$ with the minimum positive
integer  which is different from the colours assigned to the neighbours of
$v_i$ which are already coloured. For each $i$, $1\leq i\leq n-1$, the colour
assigned to $v_i$ is at most $\Delta$ since there are at most $\Delta-1$
neighbours of $v_i$ already coloured. The  colour assigned to $v_a$ is at most
$\Delta$ since $v_1$ has two neighbours (namely, $a$ and $b$) receiving the
same colour. Thus, we have a vertex-colouring of $G$ with at most $\Delta$
colours. This procedure can be implemented in $O(|V|+|E|)$ time using standard
depth-first search algorithms. 
\end{proof}

\begin{lemma}
\label{lem:FindAB}
Let $G=(V,E)$ be a biconnected graph which is not a complete graph or a cycle.
Then vertices $a$ and $b$ at distance 2 in $G$ can be found in $O(|V|+|E|)$
time such that $G\setminus\{a,b\}$ is connected.
\end{lemma}

\begin{proof}
Let $n=|V|$. Since $G$ is biconnected and not a 3-cycle, $n\geq4$.

Suppose every vertex $v$ has degree $2$ or degree $n-1$. Since $G$ is not a
cycle, at least one vertex has degree $n-1$. Thus, and since $G$ is
biconnected, at least two vertices have degree $n-1$, as otherwise $G$ would
be  the 1-connected graph shown in Figure~(a). Since $G$ is
not a complete graph there is at least one vertex of degree $2$, which implies
there is exactly two vertices of degree $n-1$; that is, $G$ is $K_{1,1,n-2}$,
as shown in Figure~(b). Since $G$ is not a 3-cycle there are
at least two vertices of degree $2$. Let $a$ and $b$ be any two degree $2$
vertices. Then $G\setminus\{a,b\}$ is connected, and we are done. Clearly this
case can be recognised and the vertices $a$ and $b$ determined in $O(|V|+|E|)$
time. 

\begin{figure}[htb]
\begin{center}
\includegraphics{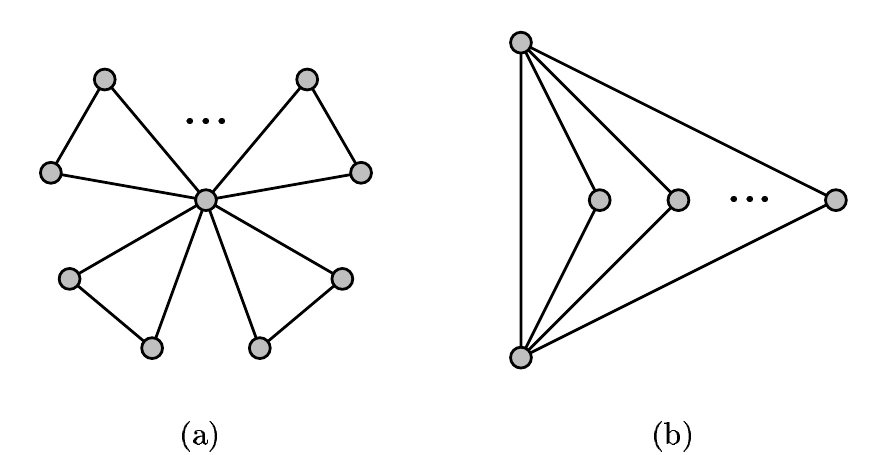}
\end{center}
\end{figure}

Otherwise $G$ has a vertex $x$ with $3\leq\deg(x)\leq n-1$. We now consider two
cases depending on the biconnectivity of $G\setminus x$. First, suppose
$G\setminus x$ is biconnected. Let $a=x$ and $b$ be any vertex at distance 2
from $x$. Then $G\setminus\{a,b\}$ is connected, and we are done. Second,
suppose $G\setminus x$ is not biconnected. Since $G$ is biconnected,
$G\setminus x$ is connected. Let $B_1$ and $B_2$ be end-blocks of $G\setminus
x$ with respective cut-points $z_1$ and $z_2$. (An end-block corresponds to a
leaf of the block-cut-tree, and since a tree has at least two leaves, $B_1$ and
$B_2$ exist.)\ Since $G$ is biconnected but $G\setminus x$ is not biconnected,
$x$ must be adjacent to vertices in $B_1$ and $B_2$ which are not $z_1$ and
$z_2$. Let $a$ and $b$ be these vertices. The only vertex adjacent to both $a$
and $b$ is $x$, and since $\deg(x)\geq3$, $G\setminus\{a,b\}$ is connected.
Again, this algorithm can be implemented in $O(|V|+|E|)$ time using depth-first
search algorithms for determining biconnectivity and the biconnected components
of a graph~\cite{Tarjan72}. \end{proof}

Of course cycles (both odd and even) and complete graphs can be recognised and
Brooks-colourings for these graphs determined in linear time. Combining this
observation with Lemmata~\ref{lem:NonBiconnected},
\ref{lem:SequentialColouring} and \ref{lem:FindAB}, we obtain the following
result.

\begin{theorem}
There is an algorithm to determine a Brooks-colouring of a given graph $G=(V,E)$ in
$O(|V|+|E|)$ time.\qed
\end{theorem}

%This algorithm has been implemented in C++ using the LEDA class libraries, and is
%publicly available at \texttt{http://www.cs.usyd.edu.au/$\sim$davidw/}.

%\bibliography{../../bib/graph,../../bib/conferences}

\def\soft#1{\leavevmode\setbox0=\hbox{h}\dimen7=\ht0\advance \dimen7
  by-1ex\relax\if t#1\relax\rlap{\raise.6\dimen7
  \hbox{\kern.3ex\char'47}}#1\relax\else\if T#1\relax
  \rlap{\raise.5\dimen7\hbox{\kern1.3ex\char'47}}#1\relax \else\if
  d#1\relax\rlap{\raise.5\dimen7\hbox{\kern.9ex \char'47}}#1\relax\else\if
  D#1\relax\rlap{\raise.5\dimen7 \hbox{\kern1.4ex\char'47}}#1\relax\else\if
  l#1\relax \rlap{\raise.5\dimen7\hbox{\kern.4ex\char'47}}#1\relax \else\if
  L#1\relax\rlap{\raise.5\dimen7\hbox{\kern.7ex
  \char'47}}#1\relax\else\message{accent \string\soft \space #1 not
  defined!}#1\relax\fi\fi\fi\fi\fi\fi} \def\cprime{$'$}

\end{document}